\documentclass[11pt]{article}

\usepackage[utf8]{inputenc}

\usepackage{amsmath,amsfonts,amsthm,amssymb,color}
\usepackage[usenames,dvipsnames,svgnames,table]{xcolor}
\definecolor{darkgreen}{rgb}{0.0,0,0.9}
\usepackage{mathtools}
\usepackage{authblk}
\usepackage{fullpage}
\usepackage{parskip}
\usepackage{comment}
\usepackage{tikz}
\usepackage{bbm}
\usepackage{dsfont}
\usepackage[sc]{mathpazo}
\usepackage[basic]{complexity}
\usepackage{algorithm2e}
\usepackage[colorlinks=true,
citecolor=OliveGreen,linkcolor=BrickRed,urlcolor=BrickRed,pdfstartview=FitH]{hyperref}
\usepackage[capitalize,nameinlink]{cleveref}
\usepackage{tcolorbox}

\newtcolorbox{wbox}
{
	colback  = white,
}

\SetKwInOut{Input}{Input}
\SetKwInOut{Output}{Output}
\SetKwFunction{Uncross}{\textsc{Uncross}}
\SetKwFunction{MergeUncross}{\textsc{MergeUncross}}
\SetKwFunction{PerfectMatching}{\textsc{PerfectMatching}}
\SetKwBlock{InParallel}{in parallel do}{end}
\SetKwFor{ParallelFor}{for}{in parallel do}{end}

\newcommand*{\suppress}[1]{}

\makeatletter
\def\thm@space@setup{%
	\thm@preskip= 10pt
	\thm@postskip=\thm@preskip 
}
\makeatother

\makeatletter
\renewcommand{\paragraph}{%
	\@startsection{paragraph}{4}%
	{\z@}{5pt}{-1em}%
	{\normalfont\normalsize\bfseries}%
}
\makeatother

\newtheorem{theorem}{Theorem}
\newtheorem{lemma}[theorem]{Lemma}
\newtheorem{corollary}[theorem]{Corollary}

\theoremstyle{definition}
\newtheorem{definition}[theorem]{Definition}
\newtheorem{remark}[theorem]{Remark}

\newtheorem{alg}[theorem]{Algorithm}

\newtheorem{example}[theorem]{Example}

\newenvironment{fminipage}%
{\begin{Sbox}\begin{minipage}}%
		{\end{minipage}\end{Sbox}\fbox{\TheSbox}}

\newcommand{\cost}{\mbox{\rm cost}}
\newcommand{\size}{\mbox{\rm size}}
\newcommand{\val}{\mbox{\rm value}}

\title{Computational Complexity of \\
the Hylland-Zeckhauser Scheme \\
for One-Sided Matching Markets}

\author[1]{Vijay V.~Vazirani\footnote{Supported in part by NSF grant CCF-1815901.}}
\author[2]{Mihalis Yannakakis}

\affil[1]{University of California, Irvine}
\affil[2]{Columbia University}

\date{}

\begin{document}
	\maketitle

\begin{abstract}
In 1979, Hylland and Zeckhauser \cite{hylland} gave a simple and general scheme for implementing a one-sided matching market using the power of a pricing mechanism. Their method has nice properties -- it is incentive compatible in the large and produces an allocation that is Pareto optimal -- and hence it provides an attractive, off-the-shelf method for running an application involving such a market. With matching markets becoming ever more prevalant and impactful, it is imperative to finally settle the computational complexity of this scheme. 

We present the following partial resolution:
\begin{enumerate}
	\item A combinatorial, strongly polynomial time algorithm for the case of $0/1$ utilities, and more generally, when each agent's utilities come from a bivalued set. 
	\item An example that has only irrational equilibria, hence proving that this problem is not in PPAD. Furthermore, its equilibria are disconnected, hence showing that the problem does not admit a convex programming formulation.
	\item A proof of membership of the problem in the class FIXP.  
\end{enumerate}

We leave open the (difficult) question of determining if the problem is FIXP-hard. 

\end{abstract}

\pagebreak
    
\section{Introduction}
\label{sec:intro}

In a brilliant and by-now classic paper, Hylland and Zeckhauser \cite{hylland} gave a simple and general scheme for implementing a one-sided matching market using the power of a pricing mechanism\footnote{See Remark \ref{rem.BM} for a discussion of the advantages of this mechanism.}. Their method is incentive compatible in the large and produces an allocation that is Pareto optimal. It can be viewed as a marriage between fractional perfect matching and a linear Fisher market, both of which admit not only polynomial time algorithms but also combinatorial ones. These facts have enticed numerous researchers over the years to seek an efficient algorithm for the Hylland-Zeckhauser (HZ) scheme. The significance of this problem has only grown in recent years, with ever more diverse and impactful matching markets being launched into our economy, e.g., see \cite{Simons}.

Our work on resolving this problem started with an encouraging sign, when we obtained a combinatorial, strongly polynomial time algorithm for the {\em unit case}, in which all utilities are 0/1, by melding a prefect matching algorithm with the combinatorial algorithm of \cite{DPSV} for the linear Fisher market, see Section \ref{sec.unit}. This algorithm can be extended to solve a more general problem which we call the {\em bivalued utilities case}, in which each agent's  utilities can take one of only two values, though the two values can be different for different agents. However, this approach did not extend any further, as described in the next section.

While studying the unit case of two-sided markets, Bogomolnaia and Moulin \cite{Bogomolnaia2004random} called it an ``important special case of the bilateral matching problem.'' They gave a number of applications, some of which are natural applications of one-sided markets as well, e.g., housemates distributing rooms, having different features, in a house. Furthermore, they say, ``Time sharing is the simplest way to deal fairly with indivisibilities of matching markets: think of a set of workers sharing their time among a set of employers.'' It turns out that the HZ (fractional) equilibrium allocation is a superior starting point for the problem of designing a randomized time-sharing mechanism; this is discussed in Remark \ref{rem.BM} after introducing the HZ model.


\subsection{The gamut of possibilities}
\label{sec.gamut}

Before presenting the rest of our results, it is worthwhile pointing out how we got to them. The most useful solution for practical applications would of course have been a combinatorial, polynomial time algorithm for the entire scheme. At the outset, this didn't seem unlikely, especially in view of the existence of such an algorithm for the unit case. Failing this, one could have sought a rational convex program, i.e., a non-linear convex program which always has a rational solution if all parameters are rational numbers \cite{va.rational}, since then interior point methods could have been employed for solving this program efficiently. However, even this approach had many stumbling blocks.

Next we considered the generalization of the bivalued utilities case to trivalued utilities, in particular, to the case of $\{0, {1 \over 2}, 1\}$ utilities. The status of this case is discussed in Section \ref{sec.discussion}. It turns out  that underlying the polynomial time solvability of a linear Fisher market is the property of weak gross substitutability\footnote{Namely, if you increase the price of one good, the demand of another good cannot decrease.}. This property is destroyed as soon as one goes to a slightly more general utility function, namely piecewise-linear, concave and separable over goods (SPLC utilities). Evidence of intractability for the latter case was established using the class PPAD introduced in \cite{PPAD}; the problem is PPAD-complete\footnote{Independently, PPAD-hardness was also established in \cite{Chen2009spending}.} \cite{VY.plc}. It turns out that equilibrium allocations for the HZ scheme do not satisfy weak gross substitutability, e.g., see Example \ref{ex.two}. 

This led us to seek a proof of PPAD-completeness for the scheme. However, a crucial requirement before embarking on such a proof is to show that there is always a rational equilibrium if all parameters of the instance are rational numbers. However, even this is not true; we found an example that admits only irrational equilibria, see Section \ref{sec.irrational}. This example consists of four agents and goods, and hence can be viewed as belonging to the four-valued utilities case; see Remark \ref{rem.irrational} for other intriguing aspects of this example. Furthermore, this example has disconnected equilibria, hence ruling out a convex programming formulation for the HZ scheme. It was then natural to turn to the class FIXP, introduced in \cite{EY07}, to establish intractability. Our proof of membership in FIXP is presented in Section \ref{sec.membership}. We leave open the problem of determining if HZ is FIXP-hard.


\subsection{Related work}
\label{sec.related}

We are aware of only the following two computational results on the HZ scheme. Using the algebraic  cell decomposition technique of \cite{Basu1995}, \cite{DK2008} gave a polynomial time algorithm for computing an equilibrium for an Arrow-Debreu market under piecewise-linear, concave (PLC) utilities (not necessarily separable over goods) if the number of goods is fixed. One can see that their algorithm can be adapted to yield a polynomial time algorithm for computing an equilibrium for the HZ scheme if the number of goods is a fixed constant. Extending these methods, \cite{Alaei2017} gave a polynomial time algorithm for the case that the number of agents is a fixed constant.

We note that there is a paucity of results showing membership in FIXP and FIXP-hardness for equilibrium problems; for a description of the class FIXP, see Section \ref{sec.FIXP}. The quintessential complete problem for this class is multiplayer Nash equilibrium \cite{EY07}.
For the case of market equilibria, in the economics literature, there are two parallel streams of results: one assumes that an excess demand function is given and the other assumes a specific class of utility functions. We provide below all FIXP-based results we are aware of for both streams.
 
\cite{EY07} proved FIXP-completeness of Arrow-Debreu markets whose excess demand functions are algebraic. Hence this result is for the first stream and it does not establish FIXP-completeness of Arrow-Debreu markets under any specific class of utility functions. Results for the second stream include proofs of membership in FIXP for Arrow-Debreu markets under Leontief and piecewise-linear concave (PLC) utility functions in \cite{Yannakakis} and \cite{Garg2016dichotomies}, respectively. This was followed by a proof of FIXP-hardness for Arrow-Debreu markets with Leontief and PLC utilities \cite{Garg2017settling}. For the case of Arrow-Debreu markets with CES (constant elasticity of substitution) utility functions, \cite{Chen2017CES} show membership in FIXP but leave open FIXP-hardness.\footnote{Computing approximate equilibria for CES
markets is PPAD-complete \cite{Chen2017CES}.}

In recent years, several researchers have proposed Hylland-Zeckhauser-type mechanisms for a number of applications, e.g., see \cite{Budish2011combinatorial, He2018pseudo, Le2017competitive, Mclennan2018efficient}. The basic scheme has also been generalized in several different directions, including two-sided matching markets, adding quantitative constraints, and to the setting in which agents have initial endowments of goods instead of money, see  \cite{Echenique2019constrained, Echenique2019fairness}.


	

\section{The Hylland-Zeckhauser Scheme}
\label{sec.prob} 

Hylland and Zeckhauser \cite{hylland} gave a general mechanism for a one-sided matching market using the power of a pricing mechanism. Their formulation is as follows: Let $A = \{1, 2, \ldots n\}$ be a set of $n$ agents and $G = \{1, 2, \ldots, n\}$ be a set of $n$ indivisible goods. The mechanism will allocate exactly one good to each agent and will have the following two properties:
\begin{itemize}
	\item The allocation produced is Pareto optimal.
	\item The mechanism is incentive compatible.
\end{itemize}

The Hylland-Zeckhauser scheme is a marriage between linear Fisher market and fractional perfect matching. The agents will reveal to the mechanism their desires for the goods by stating their von Neumann-Mogenstern utilities. Let $u_{ij}$ represent the utility of agent $i$ for good $j$. We will use language from the study of market equilibria to describe the rest of the formulation. For this purpose, we next define the linear Fisher market model. 

A {\em linear Fisher market} consists of a set $A = \{1, 2, \ldots n\}$ of $n$ agents and a set $G = \{1, 2, \ldots, m\}$ be a set of $m$ infinitely divisible goods. By fixing the units for each good, we may assume without loss of generality that there is a unit of each good in the market. Each agent $i$ has money $m_i$ and utility $u_{ij}$ for a unit of good $j$. If $x_{ij}, \ 1 \leq j \leq m$ is the {\em bundle of goods allocated to $i$}, then the utility accrued by $i$ is $\sum_j {u_{ij} x_{ij}}$. Each good $j$ is assigned a non-negative price, $p_j$. Allocations and prices, $x$ and $p$, are said to form an {\em equilibrium} if each agent obtains a utility maximizing bundle of goods at prices $p$ and the {\em market clears}, i.e., each good is fully sold and all money of agents is fully spent.

In order to mold the one-sided market into a linear Fisher market, the HZ scheme renders goods  divisible by assuming that there is one unit of probability share of each good. An {\em allocation} to an agent is a collection of probability shares over the goods. Let $x_{ij}$ be the probability share that agent $i$ receives of good $j$. Then, $\sum_j {u_{ij} x_{ij}}$ is the {\em expected utility} accrued by agent $i$. Each good $j$ has price $p_j \geq 0$ in this market and each agent has 1 dollar with which it buys probability shares. The entire allocation must form a {\em fractional perfect  matching in the complete bipartite graph} over vertex sets $A$ and $G$ as follows: there is one unit of probability share of each good and the total probability share assigned to each agent also needs to be one unit. Subject to these constraints, each agent should buy a utility maximizing bundle of goods. By definition, under such an allocation and prices, the market clears. We will define these to be an {\em equilibrium allocation and prices}; we state this formally below after giving some preliminary definitions.

\begin{definition}
Let $x$ and $p$ denote arbitrary allocations and prices of goods. By {\em size, cost and value} of agent $i$'s bundle we mean 
\[ \sum_{j \in G} {x_{ij}}, \ \ \ \ \sum_{j \in G} {p_j x_{ij}} \ \ \ \ \mbox{and} \ \ \ \ \sum_{j \in G} {u_{ij} x_{ij}}, \] 
respectively. We will denote these by $\size(i)$, $\cost(i)$ and $\val(i)$, respectively.
\end{definition}

\begin{definition}
(Hylland and Zeckhauser \cite{hylland})
	Allocations and prices $(x, p)$ form an {\em equilibrium} for the one-sided matching market stated above if:
	\begin{enumerate}
	\item  The total probability share of each good $j$ is 1 unit, i.e., $\sum_i {x_{ij}} = 1$.
	\item  The size of each agent $i$'s allocation is 1, i.e., $\size(i) = 1$.
	\item The budget of each agent is 1 dollar and price of each good is non-negative.
	\item Subject to these constraints, each agent $i$ maximizes her expected utility, i.e., maximize $\val(i)$, subject to $\size(i) = 1$ and $\cost(i) \leq 1$.
	\end{enumerate}
	\end{definition}

Using Kakutani's fixed point theorem, the following is shown:

\begin{theorem}
\label{thm.HZ}
	[Hylland and Zeckhauser \cite{hylland}] Every instance of the one-sided market defined above admits an equilibrium. 
\end{theorem}

Note that unlike the linear Fisher market, in which at equilibrium each agent $i$ must spend her money $m_i$ fully, in the HZ scheme, $i$ need not spend her entire dollar. Pareto optimality of the allocation follows from the well-known fact that allocations made at market equilibrium are Pareto optimal. Finally, if this ``market'' is large enough, no individual agent will be able to improve her allocation by misreporting utilities nor will she be able to manipulate prices. For this reason, the HZ scheme is Pareto optimal in the large.

As stated above, Hylland and Zeckhauser transform the original matching market problem, involving indivisible goods, to fractional assignments of probability shares of the goods, in order to render them divisible and bring to bear the vast machinery of market equilibria on the problem. However, their interest is still in solving the original problem, and to this end, they view each agent's allocation as a lottery over goods. In this viewpoint, agents accrue utility in an {\em expected sense} from their allocations. Once these lotteries are resolved in a manner faithful to the probabilities, an assignment of indivisible goods will result. As is well known, since the assignment is Pareto optimal {\em ex ante}, it will also be Pareto optimal {\em ex post}. Rather than executing the lotteries in a banal manner, Hylland and Zeckhauser give an elegant randomized procedure for rounding which uses randomness more efficiently.

\begin{remark}
	\label{rem.BM}
	In their paper studying the unit case of two-sided matching markets, Bogomolnaia and Moulin \cite{Bogomolnaia2004random} state that the preferred way of dealing with indivisibilities inherent in matching markets is to resort to time sharing using randomization. Their method builds on the Gallai-Edmonds decomposition of the underlying bipartite graph; this classifies vertices into three categories: disposable, over-demanded and perfectly matched. This is a much more coarse insight into the demand structure of vertices than that obtained via the HZ equilibrium. The latter is the output of a market mechanism in which equilibrium prices reflect the relative importance of goods in an accurate and precise manner, based on the utilities declared by buyers, and equilibrium allocations are as equitable as possible across buyers. Hence the latter yields a more fair and desirable randomized time-sharing mechanism.  
\end{remark}

\section{Properties of Optimal Allocations and Prices}
\label{sec:optimal}
	 
Let $p$ be given prices which are not necessarily equilibrium prices. An optimal bundle for agent $i$, $x_i$, is a solution to the following LP, which has two constraints, one for size and one for cost.

	\begin{align*}
		\max \quad & \sum_{j} x_{ij} u_{ij} \\
		\text{s.t.} \quad & \\ 
		\quad & \sum_j x_{ij} = 1 \\
		& \sum_j x_{ij} p_j \leq 1 \\
		\forall j \quad & x_{ij} \geq 0 
	\end{align*}

Taking $\mu_i$ and $\alpha_i$ to be the dual variables corresponding to the two constraints, we get the dual LP:
 
\begin{align*}
	\min \quad & \alpha_{i} + \mu_{i} \\
	\text{s.t.} \quad & \\ 
	\forall i,j \quad & \alpha_i p_j + \mu_i \geq u_{ij} \\
	      & \alpha_{i} \geq 0 
\end{align*}

Clearly $\mu_i$ is unconstrained. $\mu_i$ will be called the {\em offset} on $i$'s utilities. By complementary slackness, if $x_{ij}$ is positive then $\alpha_i p_j = {u_{ij} - \mu_i}$.   All goods $j$ satisfying this equality will be called {\em optimal goods for agent $i$}. The rest of the goods, called {\em suboptimal}, will satisfy $\alpha_i p_j > {u_{ij} - \mu_i}$. Obviously an optimal bundle for $i$ must contain only optimal goods. 

The parameter $\mu_i$ plays a crucial role in ensuring that $i$'s optimal bundle satisfies both size and cost constraints. If a single good is an effective way of satisfying both size and cost constraints, then $\mu_i$ plays no role and can be set to zero. However, if different goods are better from the viewpoint of size and cost, then $\mu_i$ attains the right value so they both become optimal and $i$ buys an appropriate combination. We provide an  example below to illustrate this. 

\begin{example}
\label{ex.one}
Suppose $i$ has positive utilities for only two goods, $j$ and $k$, with $u_{ij} = 10, \ u_{ik} = 2$ and their prices are $p_j = 2, \ p_{k} = 0.1$. Clearly, neither good satisfies both size and cost constraints optimally: good $j$ is better for the size constraint and $k$ is better for the cost constraint. If $i$ buys one unit of good $j$, she spends 2 dollars, thus exceeding her budget. On the other hand, she can afford to buy 10 units of $k$, giving her utility of 20; however, she has far exceeded the size constraint. It turns out that her optimal bundle consists of $9/19$ units of $j$ and $10/19$ units of $k$; the costs of these two goods being $18/19$ and $1/19$ dollars, respectively. Clearly, her size and cost constraints are both met exactly.  Her total utility from this bundle is $110/19$. It is easy to see that $\alpha_i = 80/19$ and $\mu_i = 30/19$, and for these settings of the parameters, both goods are optimal.
\end{example}

We next show that equilibrium prices are invariant under the operation of {\em scaling} the difference of prices from 1.

\begin{lemma}
\label{lem.scale}
Let $p$ be an equilibrium price vector and fix any $r > 0$. Let $p'$ be such that $\forall j \in G$, $p'_j - 1 = r(p_j -1)$. Then $p'$ is also an equilibrium price vector.
\end{lemma}

\begin{proof}
	Consider an agent $i$. Clearly, $\sum_{j \in G} {p_j x_{ij}} \leq 1$. Now,
	\[ \sum_{j \in G} {p'_j x_{ij}} \ = \ \sum_{j \in G} {(rp_j -r + 1) x_{ij}} \leq 1, \]
	where the last inequality follows by using $\sum_{j \in G} {x_{ij}} = 1$.
\end{proof}

Using Lemma \ref{lem.scale}, it is easy to see that if the allocation $x$ provides optimal bundles to all agents under prices $p$ then it also does so under $p'$. In the rest of this paper we will enforce that the minimum price of a good is zero, thereby fixing the scale. Observe that the main goal of the Hylland-Zeckhauser scheme is to yield the ``correct'' allocations to agents; the prices are simply a vehicle in the market mechanism to achieve this. Hence arbitrarily fixing the scale does not change the essential nature of the problem. Moreover, setting the minimum price to zero is standard \cite{hylland} and can lead to simplifying the equilibrium computation problem as shown in Remark \ref{rem.mu}. 

\begin{remark}
\label{rem.mu}
We remark that on the one hand, the offset $\mu_i$ is a key enabler in construing optimal bundles, on the other, it is also a main source of difficulty in computing equilibria for the HZ scheme. We identify here an interesting case in which $\mu_i = 0$ and this difficulty is mitigated. In particular, this holds for all agents in the unit case presented in Section \ref{sec.unit}.
	Suppose good $j$ is optimal for agent $i$, $u_{ij} = 0$ and $p_j = 0$, then it is easy to check that $\mu_i = 0$. If so, the optimal goods for $i$ are simply the maximum bang-per-buck goods; the latter notion is replete in market equilibrium papers, e.g., see \cite{DPSV}.
\end{remark}

Finally, we extend Example \ref{ex.one} to illustrate that optimal allocations for the Hylland-Zeckhauser model do not satisfy the weak gross substitutes condition in general.

\begin{example}
	\label{ex.two}
In Example \ref{ex.one}, let us raise the price of $k$ to 0.2 dollars. Then, the demands for $j$ and $k$ become $4/9$ and $5/9$, respectively. Notice that the demand for $j$ went down from $9/19$ to $4/9$.  One way to understand this change is as follows: Let us start with the old allocation of $10/19$ units of $k$. Clearly, the price of this allocation of $k$ goes up from $1/19$ to $2/19$, leaving only $17/19$ dollars for $j$. Therefore size of $j$ needs to be reduced to $17/38$. However, now the sum of the sizes becomes $37/38$, i.e., less than a unit. We wish to increase this to a unit while still keeping cost at a unit. The only way of doing this is to sell some of the more expensive good and use the money to buy the cheaper good. This is the reason for the  decrease in demand of $j$.
\end{example}

	\section{Strongly Polynomial Algorithm for Bivalued Utilities Case}
\label{sec.unit}

In this section, we will study the restriction of the HZ scheme to the bivalued utilities case, which is defined as follows: for each agent $i$, we are given a set $\{a_i, b_i \}$, where $0 \leq a_i < b_i$, and the utilities $u_{ij}, \ \forall j \in G$, are picked from this set. However first, using a perfect matching algorithm and the combinatorial algorithm \cite{DPSV} for linear Fisher markets, we will give a strongly polynomial time algorithm for the unit case, i.e., when all utilities $u_{ij}$ are $0/1$. Next we define the notion of equivalence of utility functions and show that the bivalued utilities case is equivalent to the unit case, thereby extending the unit case algorithm to this case.

We need to clarify that we will not use the main algorithm from \cite{DPSV}, which uses the notion of balanced flows and $\l_2$ norm to achieve polynomial running time. Instead, we will use the ``simple algorithm'' presented in Section 5 in \cite{DPSV}. Although this algorithm is not proven to be efficient, the simplified version we define below, called Simplified DPSV Algorithm, is efficient; in fact it runs in strongly polynomial time, unlike the balanced-flows-based algorithm of \cite{DPSV}. Remark \ref{rem.mu} provides an insight into what makes the unit case computationally easier.

{\bf Notation:}
We will denote by $H = (A, G, E)$ be the bipartite graph on vertex sets $A$ and $G$, and edge set $E$, with $(i, j) \in E$ iff $u_{ij} = 1$. For $A' \subseteq A$ and $G' \subseteq G$, we will denote by $H[A', G']$ the restriction of $H$ to vertex set $A' \cup G'$. If $\nu$ is a matching in $H$, $\nu \subseteq E$, and $(i, j) \in \nu$ then we will say that $\nu(i) = j$ and $\nu(j) = i$. For any subset $S \subseteq A$ ($S \subseteq G$), $N(S)$ will denote the set of neighbors, in $G$ ($A$), of vertices in $S$.

If $H$ has a perfect matching, the matter is straightforward as stated in Steps 1a and 1b; allocations and prices are clearly equilibrium. For Step 2, we need the following lemma.

\begin{lemma}
\label{lem.expand}
The following hold:

\begin{enumerate}
	\item For any set $S \subseteq A_2$, $|N(S)| \geq |S|$.
	\item For any set $S \subseteq G_1$, $|N(S) \cap A_1| \geq |S|$.
\end{enumerate}
\end{lemma}

\begin{proof}
1).	If $|N(S)| < |S|$ then $(G_1 \cup N(S)) \cup (A_2 - S)$ is a smaller vertex cover for $H$, leading to a contradiction.  

2). If $|N(S) \cap A_1| < |S|$ then $(G_1 - S) \cup (A_2 \cup N(S))$ is a smaller vertex cover for $H$, leading to a contradiction.  
\end{proof}

The first part of Lemma \ref{lem.expand} together with Hall's Theorem implies that a maximum matching in $H[A_2, G_2]$ must match all agents. Therefore in Step 2a, each agent $i \in A_2$ is allocated one unit of a unique good from which it derives utility 1; clearly, this is an optimal bundle for $i$. The number of goods that will remain unmatched in $G_2$ at the end of this step is $|G_2| - |A_2|$.

\bigskip

\setcounter{figure}{1} 

\begin{figure}

	\begin{wbox}
		\begin{alg}
		\label{alg.one}
		{\bf Algorithm for the Unit Case}\\
		\\
		\begin{enumerate}
			\item If $H$ has a perfect matching, say $\nu$, then do:
			 \begin{enumerate}
			 	\item $\forall i \in A$: allocate good $\nu(i)$ to $i$.
			 	\item $\forall j \in G$: $p_j \leftarrow 0$.  \ Go to Step 3.
			 \end{enumerate}
		
			\item Else do:
			\begin{enumerate}
			\item Find a minimum vertex cover in $H$, say $G_1 \cup A_2$, where $G_1 \subset G$ and $A_2 \subset A$. Let $A_1 = A - A_2$ and $G_2 = G - G_1$. 
			\item Find a maximum matching in $H[A_2, G_2]$, say $\nu$. 
			\item $\forall i \in A_2$: allocate good $\nu(i)$ to $i$.
			 	\item $\forall j \in G_2$: $p_j \leftarrow 0$.
			\item Run the Simplified DPSV Algorithm on agents $A_1$ and goods $G_1$.
			\item $\forall i \in A_1$: Allocate unmatched goods of $G_2$ to satisfy the size constraint.
			\end{enumerate}
			\item 	Output the allocations and prices computed and Halt.

		\end{enumerate} 
		\bigskip
		\end{alg}
	\end{wbox}
\end{figure} 

Allocations are computed for agents in $A_1$ as follows. First, Step 2e uses the Simplified DPSV Algorithm, which we describe below, to compute equilibrium allocations and prices for the submarket consisting of agents in $A_1$ and goods in $G_1$. At the end of this step, the money of each agent in $A_1$ is exhausted; however, her allocation may not meet the size constraint. To achieve the latter, Step 2f allocates the unmatched zero-priced goods from $G_2$ to agents in $A_1$. Clearly, the total deficit in size is $|A_1| - |G_1|$. Since this equals $|G_2| - |A_2|$, the market clears at the end of Step 2f. As shown in Lemma \ref{lem.optimal}, each agent in $A_1$ also gets an optimal bundle of goods.

Let $p$ be the prices of goods in $G_1$ at any point in this algorithm. As a consequence of the second part of Lemma \ref{lem.expand}, the equilibrium price of each good in $G_1$ will be at least 1. The Simplified DPSV algorithm will initialize prices of goods in $G_1$ to 1 and declare all goods active. The algorithm will always raise prices of active goods uniformly\footnote{In \cite{DPSV}, prices of active goods are raised multiplicatively, which amounts to raising prices of active goods uniformly for our simplified setting.}. 

For $S \subseteq G_1$ let $p(S)$ denote the sum of the equilibrium prices of goods in $S$. A key notion from \cite{DPSV} is that of a tight set; set $S \subseteq G_1$ is said to be {\em tight} if $p(S) = |N(S)|$, the latter being the total money of agents in $A_1$ who are interested in goods in $S$. If set $S$ is tight, then the local market consisting of goods in $S$ and agents in $N(S)$ clears. To see this, one needs to use the flow-based procedure given in \cite{DPSV} to show that each agent $i \in N(S)$ can be allocated 1 dollar worth of those goods in $S$ from which it accrues unit utility. Thus equilibrium has been reached for goods in $S$. 

As the algorithm raises prices of all goods in $G_1$, at some point a set $S$ will go tight. The algorithm then {\em freezes} the prices of its goods and removes them from the active set. It then proceeds to raise the prices of currently active goods until another set goes tight, and so on, until all goods in $G_1$ are frozen.

We can now explain in what sense we need a ``simplified'' version of the DPSV algorithm. Assume that as some point, $S \subset G_1$ is frozen and goods in $G_1 - S$ are active and their prices are raised. As this happens, agents in $A_1 - N(S)$ start preferring goods in $S$ relative to those in $G_1 - S$. In the general case, at some point, an agent  $i \in (A_1 - N(S))$ will prefer a good $j \in S$ as much as her other preferred goods. At this point, edge $(i, j)$ is added to the active graph. As a result, some set $S' \subseteq S$, containing $j$, will not be tight anymore and will be unfrozen. However, in our setting, the utilities of agents in $(A_1 - N(S))$ for goods in $S$ is zero, and therefore no new edges are introduced and tight sets never become unfrozen. Hence the only events of the Simplified DPSV Algorithm are raising of prices and freezing of sets. Clearly, there will be at most $n$ freezings. One can check details in \cite{DPSV} to see that the steps executed with each freezing run in strongly polynomial time, hence making the Simplified DPSV Algorithm a strongly polynomial  time algorithm\footnote{In contrast, in the general case, the number of freezings is not known to be bounded by a polynomial in $n$, as stated in \cite{DPSV}.}.

\begin{lemma}
\label{lem.optimal}
	Agents in $A_1$ will get optimal bundles of goods.
\end{lemma}

\begin{proof}
Assume that the algorithm freezes $k$ sets, $S_1, \ldots S_k$, in that order; the union of these sets being $G_1$. Let $p_1, p_2, \ldots p_k$ be the prices of goods in these sets, respectively. Clearly, successive freezings will be at higher and higher prices and therefore, $1 \leq p_1 < p_2 < \ldots < p_k$, and for $1 \leq j \leq k, p_j = |N(S_j)|/|S_j|$. If $i \in N(S_j)$, the  algorithm will allocate $1/p_j$ amount of goods to $i$ from $S_j$, costing 1 dollar.

By definition of neighborhood of sets, if $i \in N(S_j)$, then $i$ cannot have edges to $S_1, \ldots S_{j-1}$ and can have edges to $S_{j+1}, \ldots , S_k$. Therefore, the cheapest goods from which it accrues unit utility are in $S_j$, the set from which she gets 1 dollar worth of allocation. The rest of the allocation of $i$, in order to meet $i$'s size constraint, will be from $G_2$, which are zero-priced and from which $i$ gets zero utility. Clearly, $i$ gets a utility maximizing bundle satisfying both size and cost constraints.
\end{proof}

Since all steps of the algorithm, namely finding a maximum matching, a minimum vertex cover and running the Simplified DPSV Algorithm, can be executed in strongly polynomial time, we get:

\begin{lemma}
	\label{lem.unit}
	 Algorithm \ref{alg.one} finds equilibrium prices and allocations for the unit case of the Hylland-Zeckhauser scheme. It runs in strongly polynomial time.
	 \end{lemma}

\begin{definition}
	\label{def.eq}
Let $I$ be an instance of the HZ scheme and let the utility function of agent $i$ be $u_i = \{u_{i1}, u_{i12}, \ldots, u_{in} \}$. Then $u_i' = \{u_{i1}', u_{i12}', \ldots, u_{in}' \}$ is {\em equivalent} to $u_i$ if there are two numbers $s > 0$ and $h \geq 0$ such that for $1 \leq j \leq n, \ u_{ij}' = s \cdot u_{ij} + h$. The numbers $s$ and $h$ will be called the {\em scaling factor} and {\em shift}, respectively.
\end{definition}

\begin{lemma}
	\label{lem.eq}
	Let $I$ be an instance of the HZ scheme and let the utility function of agent $i$ be $u_i$. Let $u_i'$ be equivalent to $u_i$ and let $I'$ be the instance obtained by replacing $u_i$ by $u_i'$ in $I$. Then $x$ and $p$ are equilibrium allocation and prices for $I$ if and only if they are also for $I'$.
\end{lemma}

\begin{proof}
Let $s$ and $h$ be the scaling factor and shift that transform $u_i$ to $u_i'$.
	By the statement of the lemma, $x_i = \{x_{i1}, \ldots , x_{in} \}$ is an optimal bundle for $i$ at prices $p$ and hence is a solution to the optimal bundle LP given in Section \ref{sec:optimal}. The objective function of this LP is
	\[ \sum_{j = 1}^n {u_{ij} x_{ij}} . \]
	Next observe that the objective function of the corresponding LP for $i$ under instance $I'$ is 
		\[ \sum_{j = 1}^n {u_{ij}' x_{ij}}  \ = \ \sum_{j = 1}^n {(s \cdot u_{ij} + h) x_{ij}} \ = \ h + s \cdot \sum_{j = 1}^n {u_{ij} x_{ij}} , \]
		where the last equality follows from the fact that $\sum_{j = 1}^n { x_{ij}} = 1$. Therefore, the objective function of the second LP is obtained from the first by scaling and shifting. Furthermore, since the constraints of the two LPs are identical, the optimal solutions of the two LPs are the same. The lemma follows.
\end{proof}

Next, let $u_i$ be bivalued with the two values being $0 \leq a < b$. Obtain $u_i'$ from $u_i$ by replacing $a$ by $0$ and $b$ by $1$. Then, $u_i'$ is equivalent to $u_i$, with the shift and scaling being $a$ and $b-a$, respectively. Therefore the bivalued instance can be transformed to a unit instance, with both having the same equilibria. Now using Lemma \ref{lem.unit} we get: 

\begin{theorem}
	\label{thm.bivalued}
	There is a strongly polynomial time algorithm for the bivalued utilities case of the Hylland-Zeckhauser scheme.
\end{theorem}

\section{Characterizing Optimal Bundles}
\label{sec:optimal}
	 
In this section we give a characterization of optimal bundles for an agent at given prices $p$ which are not necessarily equilibrium prices. This characterization will be used critically in Section \ref{sec.membership} and to some extent in Section \ref{sec.irrational}.

{\bf Notation:} For each agent $i$, let $G_i^* \subseteq G$ denote the set of goods from which $i$ derives maximum utility, i.e., $G_i^* = \arg \max_{j \in G} \{u_{ij} \}$. With respect to an allocation $x$, let $B_i = \{ j \in G \ | \ x_{ij} > 0 \}$, i.e., the set of goods in $i$'s bundle.

We identify the following four types of optimal bundles.

{\bf Type A bundles:} $\alpha_i = 0$ and $\cost(i) < 1$.

By complementary slackness, optimal goods will satisfy $u_{ij} = \mu_i$ and suboptimal goods will satisfy $u_{ij} < \mu_i$. Hence the set of optimal goods is $G_i^*$ and $B_i \subseteq G_i^*$. Note that the prices of goods in $B_i$ can be arbitrary, as long as $\cost(i) < 1$.

{\bf Type B bundles:} $\alpha_i = 0$ and $\cost(i) = 1$.  

The only difference from the previous type is that $\cost(i)$ is exactly 1. The reason for distinguishing the two types will become clear in Section \ref{sec.membership}.

{\bf Type C bundles:}  $\alpha_i > 0$ and all optimal goods for $i$ have the same utility.

Recall that good $j$ is optimal for $i$ if\footnote{Note that under this case, optimal goods are not necessarily maximum utility goods; the latter may be suboptimal because their prices are too high.} $\alpha_i p_j = {u_{ij} - \mu_i}$. Suppose goods $j$ and $k$ are both optimal. Then $u_{ij} = u_{ik}$ and $\alpha_i p_j = {u_{ij} - \mu_i} = {u_{ik} - \mu_i} = \alpha_i p_k$, i.e., $p_j = p_k$. Since $\alpha_i > 0$, by complementary slackness, $\cost(i) = 1$. Further, since $\size(i) = 1$, we get that each optimal good has price 1.

{\bf Type D bundles:} $\alpha_i > 0$ and not all optimal goods for $i$ have the same utility.

Suppose goods $j$ and $k$ are both optimal and $u_{ij} \neq u_{ik}$. Then $\alpha_i p_j = {u_{ij} - \mu_i} \neq {u_{ik} - \mu_i} = \alpha_i p_k$, i.e., $p_j \neq p_k$. Therefore optimal goods have at least two different prices.  Since $\alpha_i > 0$, by complementary slackness, $\cost(i) = 1$. Further, since $\size(i) = 1$, there must be an optimal good with price more than 1 and an optimal good with price less than 1. Finally, if good $z$ is suboptimal for $i$, then $\alpha_i p_z < u_{iz} - \mu_i$.


\section{An Irrational Example}
\label{sec.irrational}

Our example has 4 agents $A_1, \ldots, A_4$ and 4 goods $g_1, \ldots, g_4$\footnote{It is easy to see, by analyzing relations in the bipartite graph on agents and goods with edges corresponding to non-zero allocations, that any instance with 3 agents and 3 goods and rational utilities has a rational equilibrium.}.
The agents' utilities for the goods are given in Table \ref{tab:ctexample}, with rows correspond to  agents and columns to goods. The quantity $-\infty$ stands for 
a sufficiently large negative number $-M$ ($-M \leq -20$ suffices).
If desired, we can shift up the utilities so that they are nonnegative.

\begin{table}[h!]
  \begin{center}
    \caption{Agents' utilities.}
    \label{tab:ctexample}
    \begin{tabular}{c | c c c c} 
       & $g_1$ & $g_2$ & $g_3$ & $g_4$\\
     \hline
      $A_1$ & 10 & 20 & $- \infty$ & 40 \\
      $A_2$ & 10 & 15 & $- \infty$ & 40\\
      $A_3$ & 10 & $-\infty$ & 30 & $- \infty$ \\
      $A_4$ & $-\infty$ & 20 & 30 & $- \infty$\\
    \end{tabular}
  \end{center}
\end{table}

Thus, agents $A_1$ and $A_2$ like, to varying degrees, three goods only,
$g_1, g_2, g_4$, while agents $A_3$ and $A_4$ like two goods each, $\{ g_1, g_3\}$ and  $\{ g_2, g_3\}$, respectively. The precise values of the utilities are not important;
the important aspects are: which goods each agent likes, the order between them, and the
ratios $\frac{u_{14} - u_{12}}{u_{12} - u_{11}}$
and $\frac{u_{24} - u_{22}}{u_{22} - u_{21}}$. Notice that the latter are unequal.

We will show that this example has two equilibrium solutions.
In both solutions, agents $A_3$ and $A_4$ buy the two goods that they like.
In one solution, $A_1$ buys only the goods $g_2$ and $g_4$, and $A_2$
buys all three goods $g_1, g_2, g_4$,
while in the other solution, $A_1$ buys all three goods $g_1, g_2, g_4$
and $A_2$ buys $g_1$ and $g_4$ only.
In both solutions, the minimum price of a good is 0 and the
price of the other three goods is irrational.

\begin{lemma}
\label{lem.L1}
Equilibrium prices satisfy:
\[ 0=p_1 < p_2 < 1 \ \ \mbox{and} \ \ p_3, \ p_4 >1 . \]
The equilibrium bundle of each agent is of Type D and contains goods having positive utilities only.
\end{lemma}

\begin{proof}
Suppose $p_3 \leq 1$. Then agents $A_3$ and $A_4$ will
demand 1 unit each of good $g_3$, leading to a contradiction.
Similarly, if $p_4 \leq 1$ then $A_1$ and $A_2$ will demand
1 unit each of $g_4$. Therefore, $p_3, p_4 >1$.
Since the maximum utility goods of every agent have price $>1$,
all agents spend exactly 1.
Therefore, the sum of the prices of the goods is 4.

Suppose $p_2 =0 \leq p_1$.
Then $A_1, A_2, A_4$ do not buy $g_1$, since they prefer $g_2$ and it is weakly cheaper than $g_1$. Therefore $A_3$ must buy the entire unit of $g_1$.
Clearly $A_1, A_2$ do not buy $g_3$, since they prefer $g_2$. Therefore,
the only agent who buys $g_3$ is $A_4$; however, she cannot
afford the entire unit of $g_3$ since $p_3>1$, contradicting market clearing. Therefore
$p_2 >0$ and hence the 0-priced good is $g_1$ and
$p_1 =0 < p_2$. Furthermore, $p_2 + p_3+p_4 =4$.

Next suppose $p_2 \geq 3/4$. Then $p_4 =4 -(p_2+p_3) < 9/4$. For both agents $A_1$
and $A_2$, a combination of $g_1$ and $g_4$ in proportion 2:1 has
a price less than $3/4$ for one unit and utility 20, and is therefore preferable to $g_2$.
Hence, $A_1, A_2$ will not buy any $g_2$, and since
$A_3$ does not buy any $g_2$ either, since she prefers $g_1$,
it follows that $A_4$ must buy the entire unit of $g_2$.
This is possible only if $p_2=1$ and $A_4$
buys nothing else; in particular, she does not buy any $g_3$. 
Clearly, $A_1, A_2$ do not buy any $g_3$ since they prefer $g_1$.
Therefore the entire unit of $g_3$ must be bought by $A_3$,
which is impossible because $p_3 >1$. Hence $p_2 < 3/4$.
These facts together with $p_1=0 < p_2 < 1 < p_3,p_4$ imply that the agents' bundles are not Type B or C. Therefore they are all of Type D.

Finally we prove that none of the agents will buy an undesirable good. For $A_1, A_2, A_3$,
such a good is dominated by another lower-priced good. Since $p_4 >1$, $A_4$ does not buy $g_4$.
Suppose agent $A_4$ buys good $g_1$. Since she spends 1 dollar, she must also buy $g_3$.
Therefore we have:
$\alpha_4 p_1 + \mu_4 = u_{41} =-M$, where $M \geq 2)$. Therefore $ \mu_4 = -M$. Also
$\alpha_4 p_3 + \mu_4 = u_{43} = 30$, therefore $\alpha_4 p_3  = M+30$.
Furthermore,  $\alpha_4 p_2 + \mu_4 \geq u_{42}= 20$,
hence $p_2 \geq (M+20)/(M+30)$, which contradicts $p_2 < 3/4$.
Therefore, no agent buys any undesirable good.
\end{proof}

\medskip

\begin{lemma}
\label{lem.L2}
One of the agents $A_1, A_2$ buys all three desirable goods.
If $A_1$ buys $g_1, g_2, g_4$, then
$A_2$ buys $g_1, g_4$ only.
If $A_2$ buys $g_1, g_2, g_4$, then
$A_1$ buys $g_2, g_4$ only.
\end{lemma}

\begin{proof}
Since all the bundles are of Type D, every bundle has at least two goods; clearly,
every good is bought by at least two agents.

Suppose that every agent buys two goods and every good is bought by two agents. If so, one of $A_1, A_2$ must buy $g_1, g_4$ and the other must buy $g_2,g_4$. Consider the graph with the goods as nodes
and an edge joining two nodes if they are bought by the same agent. 
This graph must be the 4-cycle $g_1, g_4, g_2, g_3, g_1$.
Therefore for some $a$, $0 < a < 1$, each agent buys $a$ units of one good and $b = 1-a$ units of the second good and each good is sold to two agents in the amounts of $a$ and $b$.

Let $r_i = |1-p_i|$. Observe that for every edge $(g_i,g_j)$ of the cycle, one price is $<1$ and the other price is $>1$, and we have $ap_i + bp_j =1$. Therefore $a r_i - b r_j=0$,
and $\frac{r_i}{r_j}=\frac{b}{a}$.
Hence 
\[ \frac{r_1}{r_4}=\frac{r_4}{r_2}=\frac{r_2}{r_3}=\frac{r_3}{r_1} , \]
which implies that all the $r_i$ are equal. Therefore $p_1 =p_2$, contradicting the previous claim that $p_1 < p_2$.
Hence at least one of $A_1, A_2$ will buy all three of her desirable goods.

Suppose that $A_1$ buys all three desirable goods $g_1, g_2, g_4$.
Then we have $\alpha_1 p_j + \mu_1 = u_{1j}$ for $j=1,2,4$.
Therefore, $(p_4-p_1)/(p_4-p_2)= (u_{14}-u_{11})/(u_{14}-u_{12})=30/20$.
Agent $A_2$ buys $g_4$ and at least one of $g_1, g_2$.
Suppose she buys $g_2$.
Then $\alpha_2 p_j + \mu_2 = u_{2j}$ for $j=2,4$,
hence $\alpha_2(p_4 - p_2) = u_{24}-u_{22} = 25$.
This implies that $\alpha_2(p_4 - p_1) > 30 = u_{24}-u_{21}$,
hence $\alpha_2 p_1 + \mu_2 < u_{21}$, a contradiction.
Therefore $A_2$ does not buy $g_2$ and she buys $g_1$ and $g_4$ only.

Next suppose $A_2$ buys all three desirable goods $g_1, g_2, g_4$.
By a similar argument we will prove that $A_1$ buys only two goods.
We have $\alpha_2 p_j + \mu_2 = u_{2j}$ for $j=1,2,4$.
Therefore, $(p_4-p_1)/(p_4-p_2)= (u_{24}-u_{21})/(u_{24}-u_{22})=30/25$.
Agent $A_1$ buys $g_4$ and at least one of $g_1, g_2$.
Suppose that she buys $g_1$.
Then $\alpha_2 p_j + \mu_2 = u_{2j}$ for $j=1,4$,
hence $\alpha_2(p_4 - p_1) = u_{24}-u_{21} = 30$.
This implies that $\alpha_2(p_4 - p_2) > 20 = u_{14}-u_{12}$,
hence $\alpha_2 p_2 + \mu_2 < u_{12}$, a contradiction.
Therefore, $A_1$ does not buy $g_1$, hence she buys $g_2$ and $g_4$ only.
\end{proof}

\medskip

\begin{theorem}
The instance of Table \ref{tab:ctexample} has two equilibria; in both, allocations to agents and prices of goods, other than the zero-priced good, are all irrational.
\begin{enumerate}
	\item  $p_1=0, \ \ p_2= (23- \sqrt{17})/32$, \ \
$p_3 = ( 9 + \sqrt{17}) /8$, \ \
$p_4 = (69 - 3 \sqrt{17})/32$.
\item
$p_1=0, \ \ p_2= (41 -\sqrt{113})/98, \ \ 
p_3= (15 +\sqrt{113} )/14, \ \
p_4= (246 -6 \sqrt{113})/98$.
\end{enumerate}
\end{theorem}

\begin{proof}
Let $r_i = |1-p_i|$. By Lemma \ref{lem.L1}, $r_1=1$. We consider the two cases established in Lemma \ref{lem.L2}.

{\bf Case 1.} $A_1$ buys $g_1, g_2, g_4$, and
$A_2$ buys $g_1, g_4$.

Agent $A_3$ spends her dollar on goods $g_1, g_3$ in the proportion
$r_3 : r_1$, i.e., $r_3:1$.
Therefore, $x_{31} = \frac{r_3}{1+r_3}, \ x_{33}= \frac{1}{1+r_3}$.
Agent $A_4$ buys goods $g_2, g_3$ in the proportion
$r_3 : r_2$.
Therefore, $x_{42} = \frac{r_3}{r_2+r_3}, \ x_{43}= \frac{r_2}{r_2+r_3}$.
Since only agents $A_3$ and $A_4$ buy good $g_3$,
we have $x_{31} = 1-x_{33}= x_{43}$, and $x_{42} =1-x_{43}= x_{33}$.
This implies $r_3^2= r_2$ ...  (1).

Since agent $A_1$ buys $g_1, g_2, g_4$, we have,
$\frac{u_{14}-u_{12}}{u_{12}-u_{11}}=\frac{p_{14}-p_{12}}{p_{12}-p_{11}}$.
Therefore $r_2+r_4 = 2(1-r_2)$ ... (2).

Agent $A_2$ buys goods $g_1, g_4$ in the proportion $r_4: r_1$.
Therefore, $x_{21} = \frac{r_4}{1+r_4}, \ x_{24}= \frac{1}{1+r_4}$.
Hence, $x_{14}= 1-x_{24}= \frac{r_4}{1+r_4}$.

Since $g_2$ is bought by $A_1$ and $A_2$ only, we have $x_{12} = 1- x_{42} = \frac{r_2}{r_2+r_3}$.

Since the cost of $A_1$'s bundle is 1,
$p_1 x_{11} + p_2 x_{12} + p_4 x_{14}=1$.
Therefore, $(1-r_2)\frac{r_2}{r_2+r_3} + (1+r_4)\frac{r_4}{1+r_4} =1$.
Hence, $(1-r_2)r_2 + (r_4-1)(r_2+r_3) = 0$ ... (3).

Now we have three equations, (1), (2) and (3), in three unknowns $r_1, r_2, r_3$.
Using (1) and (2) we can express $r_1$ and $r_2$ in terms of $r_3$.
Letting $r_3=y$, we have from (1), $r_2 = y^2$,
and from (2), $r_4 = 2- 3 r_2 = 2 - 3 y^2$.
Substituting into (3), we get
$(1-y^2)y^2 +(1-3y^2)(y^2+y)=0$,
which simplifies to $y(y+1)(4y^2 -y -1) =0$.

The only positive solution is $y = \frac{1+\sqrt{17}}{8}$.
Therefore, 
\[p_1 = 0, \ \ \ \ p_2 = 1-r_2= 1-y^2 = \frac{23- \sqrt{17}}{32}, \ \ \ \ 
p_3 =1+r_3= 1+y = \frac{9 + \sqrt{17}}{8}, \] 
\[p_4 = 1+r_4= 3-3y^2 = \frac{69 - 3 \sqrt{17}}{32} .\]

Once we have the value of $y$, we get:
\[ r_1 = 1, \ \ \ r_2 = y^2, \ \ \ r_3 = y \ \ \ \mbox{and} \ \ \ r_4 = 2(1-y^2) - y.\]
In terms of these four values, the allocations of the agents are:
\[ A_1: \ \ x_{12} = \frac{r_2}{r_2 + r_3}, \ \ \ \ x_{13}= 1 - \frac{1}{1+r_3} - \frac{r_2}{r_2 + r_3}, \ \ \ \ r_{14} = \frac{r_4}{1+r_4} \]
\[ A_2: \ \ x_{21} = \frac{r_4}{1+r_4}, \ \ \ \ x_{24} = \frac{1}{1 + r_4} \]
\[ A_3: \ \ x_{31} = \frac{r_3}{1+r_3}, \ \ \ \ x_{33} = \frac{1}{1 + r_3} \]
\[ A_4: \ \ x_{41} = \frac{r_3}{r_2+r_3}, \ \ \ \ x_{43} = \frac{r_2}{r_2 + r_3} \]

\medskip

{\bf Case 2.} $A_2$ buys $g_1, g_2, g_4$, and
$A_1$ buys $g_2, g_4$.

The allocations for agents $A_3, A_4$ are the same as above, i.e.,
$x_{31} = \frac{r_3}{1+r_3}, \ \ x_{33}= \frac{1}{1+r_3}$,
and $x_{42} = \frac{r_3}{r_2+r_3}, \ \ x_{43}= \frac{r_2}{r_2+r_3}$.
Again we have $x_{31} = x_{43}$ and $x_{42} = x_{33}$,
which implies $r_3^2= r_2$  ... (1).

Since agent $A_2$ buys $g_1, g_2, g_4$, we have,
$\frac{u_{24}-u_{22}}{u_{22}-u_{21}}=\frac{p_{24}-p_{22}}{p_{22}-p_{21}}$,
therefore $r_2+r_4 = 5(1-r_2)$ ... (2').

$A_1$ buys goods $g_2, g_4$ in the proportion $r_4: r_2$.
Therefore, $x_{12} = \frac{r_4}{r_2+r_4}, x_{14}= \frac{r_2}{r_2+r_4}$.
Hence, $x_{24}= 1-x_{14}= \frac{r_4}{r_2+r_4}$.

Since $g_2$ is bought by $A_2$ and $A_4$ only, we have $x_{22} = 1- x_{12}-x_{42}$.

Since the cost of $A_2$'s bundle is 1,
$p_1 x_{21} + p_2 x_{22} + p_4 x_{24}=1$.
Since $p_1=0$ and $x_{24}=\frac{r_4}{r_2+r_4}=x_{12}$,
we have $p_2 (1- x_{12}-x_{42}) + p_4 x_{12}=1$,
hence $(p_4-p_2)x_{12} -p_2 x_{42} -(1-p_2)=0$.
Therefore, $(r_2+r_4) \frac{r_4}{r_2+r_4} -(1-r_2)\frac{r_3}{r_2+r_3}-r_2 =0$.
Some algebra, yields the same equation as in the previous case,
$(1-r_2)r_2 + (r_4-1)(r_2+r_3) = 0$ ... (3).

We can solve now for $r_1, r_2, r_3$.
Using (1) and (2') we can express $r_1$ and $r_2$ in terms of $r_3$.
Letting $r_3=y$, we have from (1), $r_2 = y^2$,
and from (2'), $r_4 = 5- 6 r_2 = 5 - 6 y^2$.
Substituting into (3), we get
$(1-y^2)y^2 +(4-6y^2)(y^2+y)=0$,
which simplifies to $y(y+1)(7y^2 -y -4) =0$.

The only positive solution is $y = \frac{1+\sqrt{113}}{14}$.
Therefore, 
\[p_1 = 0, \ \ \ \ p_2 = 1-r_2= 1-y^2 =\frac{41 -\sqrt{113}}{98}, \ \ \ \ 
p_3 =1+r_3= 1+y = \frac{15 +\sqrt{113}}{14}, \] 
\[p_4 = 1+r_4= 6-6y^2 = \frac{246 -6 \sqrt{113}}{98}.\]

As in the previous case, the value of $y$ gives:
\[ r_1 = 1, \ \ \ r_2 = y^2, \ \ \ r_3 = y \ \ \ \mbox{and} \ \ \ r_4 = 5 - 6y^2.\]

In terms of these four values, the allocations of the agents are:
\[ A_1: \ \ x_{12} = \frac{r_4}{r_2 + r_4}, \ \ \ \ x_{14}= \frac{r_2}{r_2 + r_4} \]
\[ A_2: \ \ x_{22} = 1 - \frac{r_4}{r_2 + r_4} - \frac{r_3}{r_2 + r_4}, \ \ \ \ x_{24} = \frac{r_4}{r_2 + r_4} \]
\[ A_3: \ \ x_{31} = \frac{r_3}{1+r_3}, \ \ \ \ x_{33} = \frac{1}{1 + r_3} \]
\[ A_4: \ \ x_{42} = \frac{r_3}{r_2+r_3}, \ \ \ \ x_{43} = \frac{r_2}{r_2 + r_3} \]

\end{proof}

Since this example has disconnected equilibria, we get:

\begin{corollary}
	The Hylland-Zeckhauser scheme does not admit a convex programming formulation.
\end{corollary}

\begin{remark}
	\label{rem.irrational}
	Observe that in both equilibria, the allocations of all four agents are irrational even though each one of them spends their dollar completely and the allocations form a fractional perfect matching, i.e., add up to 1 for each good and agent.
\end{remark}



	\section{The Class FIXP}
\label{sec.FIXP}

The class FIXP was introduced in \cite{EY07}. It is essentially the class of problems that can be cast (in polynomial time) as the problem of computing a fixed point of an algebraic Brouwer function. Recall that basic complexity classes, such as P, NP, NC and \#P, are defined via machine models. For the class FIXP, the role of a ``machine model'' is played by one of the following: a straight line program, an algebraic formula, or a circuit; further it must use the standard arithmetic operations of +, - * /, $\min$ and $\max$. These restrictions imply that the function defined by this ``machine'' is guaranteed to be continuous.

The quintessential complete problem for the class FIXP is Nash equilibrium; this was established in \cite{EY07}. A three-player game may have only irrational equilibria; this was pointed out by Nash himself \cite{Nash1951non}. Corresponding to a three-player game, Nash defined a continuous function on a closed, compact domain via an algebraic formula and showed that the fixed points of this function are exactly the equilibria of the game. Since Brouwer's Theorem guarantees existence of fixed points, each such a game must admit an equilibrium. Nash's formula uses the standard algebraic operations and $\min, \max$, which implies immediately the membership of the Nash equilibrium problem in FIXP.

The problem of computing an equilibrium for a $k$-player game, for $k \geq 3$, game is FIXP-complete  \cite{EY07}. In contrast, a two-player bimatrix game always admits rational equilibria and computing it is PPAD-complete \cite{Daskalakis2009complexity, Chen2009settling}; the class PPAD was defined in \cite{PPAD}. The two classes, PPAD and FIXP seem to be rather disparate: whereas PPAD is contained in function classes NP $\cap$ co-NP, the class FIXP lies somewhere between P and PSPACE, and is likely to be closer to the harder end of PSPACE.

Since we will establish membership in FIXP using straight line programs, we provide further details on this ``machine model'' for FIXP. Such a program should satisfy the following:
\begin{enumerate}
	\item The program does not have any conditional statements, such as if ... then ... else.
	\item It uses the standard arithmetic operations of +, - * /, $\min$ and $\max$.
	\item It never attempts to divide by zero.
\end{enumerate}

A {\em total problem} is one which always has a solution, e.g., Nash equilibrium and Hylland-Zeckhauser equilibrium. A total problem is in FIXP if there is a polynomial time algorithm which given an instance $I$ of length $|I|=n$, outputs a polynomial sized straight line program which computes function $F_I$ on a closed, convex, real-valued domain $D(n)$. Because of the restrictions imposed on the straight line program, function $F_I$ will be continuous and by Brouwer's Theorem, has at least one fixed point. We require that each fixed point of $F_I$ be a solution to instance $I$. Hence, a proof of membership in FIXP requires us to provide the domain for function $F$ and a straight line program which defines $F$.

	\section{Membership in FIXP}
\label{sec.membership}

Let $p$ and $x$ denote the allocation and price variables. We will give a function $F$ over these variables and a closed, compact, real-valued domain $D$ for $F$. The function will be specified by a polynomial length straight line program using the algebraic operations of $+, -, *, /, \min$ and $\max$, hence guaranteeing that it is continuous. We will prove that all fixed points of $F$ are equilibrium allocations and prices, hence proving that Hylland-Zeckhauser is in FIXP.

{\bf Notation:} We will denote the set $\{ 1, \ldots . n \}$ by $[n]$.  $x_i$ will denote agent $i$'s bundle. For each agent $i$, choose one good from $G_i^*$ and denote it by $i^*$. 
If $e$ is an expression, we will use $(e)_+$ as a shorthand
for $\max \{0, e \}$.

Domain $D = D_p \times D_x$, where $D_p$ and $D_x$ are the domains for $p$ and $x$, respectively, with $D_p = \{p \ | \ \forall j \in [n], p_j \in [0, n] \}$ and $D_x = \{x \ | \ \forall i \in [n], \ \sum_{j \in G} {x_{ij}} = 1, \ \mbox{and} \ \forall i,j \in [n], \ x_{ij} \geq 0 \}$.

Let $(p', x') = F(p, x)$. $(p, x)$ can be viewed as being composed of $n+1$ vectors of variables, namely $p$ and for each $i \in [n], \ x_i$. Similarly, we will view $F$ as being composed of $n+1$ functions, $F_p$ and for each $i \in [n], \ F_i$, where $p' = F_p(p, x)$ and for each $i \in [n]$, $x'_i = F_i(p, x)$. The straight line programs for $F_p$ and $F_i$ are given in Algorithm \ref{alg.p} and Algorithm \ref{alg.x}, respectively. It is easy to see that if $F_i$ alters a bundle, the new bundle still remains in the domain; in particular, $\forall i \in [n], \ \size(i) = 1$.  

{\bf Requirements on $F$:} Observe that $(p, x)$ will be an equilibrium for the market if, in addition to the conditions imposed by the domain, it satisfies the following:
\begin{enumerate}
	\item $\forall j \in [n], \ \sum_{i \in A} {x_{ij}} = 1$.
	\item $\forall i \in [n], \ \cost(i) \leq 1$.
	\item $\forall i \in [n]$, $x_i$ is an optimal bundle for $i$.
\end{enumerate}
 
Function $F$ has been constructed in such a way that if any of these conditions is not satisfied by $(p, x)$, then $F(p, x) \neq (p, x)$, i.e., $(p, x)$ is not a fixed point of $F$. Equivalently, every fixed point of $F$ must satisfy all these conditions and is therefore an equilibrium.

\bigskip

\setcounter{figure}{1} 

\begin{figure}

	\begin{wbox}
		\begin{alg}
		\label{alg.p}
		{\bf Straight line program for function $F_p$}\\
		
		\begin{enumerate}
			\item For all $j \in [n]$ do: $p_j
		             \leftarrow \min \{n, \max\{0,  p_j + \sum_{i \in A} {x_{ij}	-1} \} \}$	
		             \item  $r \leftarrow \min_{j \in [n]} \{p_j\}$
		             \item  For all $j \in [n]$ do: $p_j \leftarrow p_j - r$
			 \end{enumerate}
		\bigskip
		\end{alg}
	\end{wbox}
\end{figure} 

\bigskip
\bigskip
\bigskip
\bigskip

\setcounter{figure}{1} 

\begin{figure}

	\begin{wbox}
		\begin{alg}
		\label{alg.x}
		{\bf Straight line program for function $F_i$}\\
		
		\begin{enumerate}
			\item $r \leftarrow ( \sum_j p_j x_{ij} -1)_+$. 
			\item For all $j \in [n]$ do: 
	            $x_{ij} \leftarrow \frac{x_{ij}+r \cdot (1-p_j)_+}{1+r \cdot \sum_k (1-p_k)_+}$
	            
	            \bigskip
			\item $t \leftarrow (1 - \sum_j p_j x_{ij})_+ $
			\item  For all $ k \notin G_i^*$ do:
			\begin{enumerate}
			\item $d \leftarrow \min \{x_{ik}, \frac{t}{n^2} \}$
			\item  $x_{ik} \leftarrow x_{ik} - d$ 
			\item  $x_{ii^*} := x_{ii^*} +d$
			\end{enumerate}

			\bigskip
			\item  For all pairs $j, k$ of goods s.t. $u_{ij} \leq u_{ik}$ do:
			 \begin{enumerate}
			\item   $d \leftarrow \min \{ x_{ij}, (p_j -p_k)_+ \}$
			\item   $x_{ij} \leftarrow x_{ij} + d/n^2$
			\item   $x_{ik} \leftarrow x_{ik} - d/n^2$
			\end{enumerate}
			
			\bigskip
			\item For all triples $j, k, l$ of goods such that $u_{ij} < u_{ik} < u_{il}$ do:
		    \begin{enumerate}
		    	\item  $d \leftarrow \min \{ x_{ik}, ( (u_{il}-u_{ik})(p_k -p_j) - (u_{ik}-u_{ij})(p_l -p_k) )_+ \}$
		    	\item  $x_{ik} \leftarrow x_{ik} -d$
		    	\item  $x_{ij} \leftarrow x_{ij} + \frac{u_{il}-u_{ik}}{u_{il}-u_{ij}}d$
		    	\item  $x_{il} \leftarrow x_{il} + \frac{u_{ik}-u_{ij}}{u_{il}-u_{ij}}d$
			 \end{enumerate}
			 
			 \bigskip
			\item For all triples $j, k, l$ of goods such that $u_{ij} < u_{ik} < u_{il}$ do:
			 \begin{enumerate}
			 \item $d := \min ( x_{ij}, x_{il}, ( (u_{ik}-u_{ij})(p_l -p_k) - (u_{il}-u_{ik})(p_k -p_j) )_+ )$
			 \item  $x_{ik} := x_{ik} + d$
			 \item  $x_{ij} := x_{ij} - \frac{u_{il}-u_{ik}}{u_{il}-u_{ij}}d$
			 \item  $x_{il} := x_{il} - \frac{u_{ik}-u_{ij}}{u_{il}-u_{ij}}d$
			 \end{enumerate}
						
		\end{enumerate}
		\bigskip
		\end{alg}
	\end{wbox}
\end{figure} 

\begin{lemma}
	\label{lem.two}
	If $(p, x)$ is a fixed point of $F$, as defined in Algorithms \ref{alg.p} and \ref{alg.x}, then
	\begin{enumerate}
		\item $\forall i \in [n], \ \cost(i) \leq 1$.
		\item $\forall j \in [n], \ \sum_{i \in A} {x_{ij}} = 1$.
		\item $\exists z \in G$ such that $p_z = 0$.
	\end{enumerate}
\end{lemma}

\begin{proof}
	For the first two statements, we argue below that if any of these conditions is not satisfied then $F(p, x) \neq (p, x)$, i.e., $(p, x)$ is not a fixed point. 
	\begin{enumerate}
			\item If for some $i \in [n], \ \cost(i) > 1$, then Steps 1 and 2 of $F_i$ will modify $x_i$. The subsequent steps of $F_i$ will not restore the initial value of $x_i$. 
		\item  Suppose that there is a good $j$ such that $\sum_i x_{ij} \neq 1$.
Since $\sum_j x_{ij} =1$ for all agents $i \in [n]$,
there must be a good $k$ such that $\sum_i x_{ik} < 1$, 
and another good $l$ such that $\sum_i x_{il} > 1$.

We claim that then $p_k =0$.
Since $\sum_i x_{ik} < 1$, if $p_k > 0$, then line 1 will strictly decrease $p_k$, and line 3 certainly does not increase it, contradicting $F_p(p,x) =p$.
Thus, $p_k=0$, the price $p_k$ will stay $0$ after line 1, hence 
$r=0$ in line 2, and line 3 will not change any prices.

On the other hand, we claim that $p_l =n$.
Since $\sum_i x_{il} > 1$, if $p_l < n$, then line 1 will increase strictly $p_l$, and since line 3 has no effect, this contradicts
$F_p(p,x) =p$.

But $\cost(i)= \sum_j p_j x_{ij} \leq 1$ for all $i \in [n]$ implies that
$\sum_i \sum_j p_j x_{ij} \leq n$, which contradicts
the fact that $p_l =n$ and $\sum_i x_{il} > 1$, hence 
$\sum_i p_l x_{il} > n$.

		\item Steps 2 and 3 of $F_p$ ensure that there is a good having price 0.
	\end{enumerate}
\end{proof} 
		
We will next prove that at a fixed point $(p, x)$, $x_i$, must be an optimal bundle for $i$. Before embarking on this, we need to prove that if $(p, x)$ is a fixed point with a zero-priced good, then no step of $F$ will change $(p, x)$, i.e., it couldn't be that some step(s) of $F$ change $(p, x)$ and some other step(s) change it back, restoring it to $(p, x)$. It is easy to check that $F_p$ will make no change. The next lemma uses a potential function argument to prove that $F_i$ will not change $x_i$. 

\begin{lemma}
	\label{lem.potential}
	Let $(p, x)$ be such that $p \in D_p$, $x \in D_x$ and $\exists z \in G$ such that $p_z = 0$. Then, only the following two possibilities hold about the steps taken by $F_i$: 
	\begin{enumerate}
		\item $\cost(i) < 1$ and $F_i$ strictly increases $\val(i)$ while maintaining $\cost(i) < 1$. 
		\item $F_i$ weakly decreases $\cost(i)$ and weakly increases $\val(i)$, with at least one change being strict.
	\end{enumerate}
\end{lemma}

\begin{proof}
For the first part, observe that if $\cost(i) < 1$, the only step that can kick in is Step 3. If so, Step 4 maintains $\cost(i) < 1$ and increases $\val(i)$.

For the second part, we consider the following four cases.

\begin{enumerate}
	\item Since $\cost(i) \leq 1$, Step 1 will set $r$ to 0. As a result Step 2 will not make any change to $x_i$.
	\item Step 5 weakly increases $\val(i)$ and strictly decreases $\cost(i)$.
	\item If Step 6 kicks in, then the net change in $\val(i)$ is
\[d \frac{u_{il}-u_{ik}}{u_{il}-u_{ij}}u_{ij} + d \frac{u_{ik}-u_{ij}}{u_{il}-u_{ij}}u_{il} -d u_{ik} =0. \] 
The net change in $\cost(i)$ is
\[d \frac{u_{il}-u_{ik}}{u_{il}-u_{ij}}p_{j} + d \frac{u_{ik}-u_{ij}}{u_{il}-u_{ij}}p_{l} -d p_{k}
= \frac{\Delta}{u_{il}-u_{ij}}, \]
where $\Delta = (u_{ik}-u_{ij})(p_l -p_k)- (u_{il}-u_{ik})(p_l -p_k) < 0$.

\item
If Step 7 kicks in, then the net change in $\val(i)$ is again 0, and the
net change in cost is $\frac{  \Delta}{u_{il}-u_{ij}} < 0$.
\end{enumerate}
\end{proof}

\begin{corollary}
	\label{cor.potential}
	If $(p, x)$ is a fixed point with a zero-priced good, then no step of $F$ will change $(p, x)$.
\end{corollary}

Next, we observe that if the set of distinct utilities of $i$, $\{u_{ij} \ | \ j \in G \}$, is a singleton, then any bundle satisfying $\size(i) = 1$ and $\cost(i) \leq 1$ is optimal. Henceforth we will assume that this set has at least two elements.

\begin{lemma}
	\label{lem.bundle}
	If $(p, x)$ is a fixed point of $F$, as defined in Algorithms \ref{alg.p} and \ref{alg.x}, then $x_i$ is an optimal bundle for $i$ at prices $p$.
\end{lemma}
		
\begin{proof}
We will consider the following exhaustive list of cases. Each contradiction is based on applying Corollary \ref{cor.potential}. We will assume that $\alpha_i$ and $\mu_i$ are the optimal variables of the dual to $i$'s optimal bundle LP and that $u = \max_j \{u_{ij} \}$.

{\bf Case 1:} Assume that $\cost(i) < 1$. If $B_i \not \subseteq G_i^*$, then Steps 3 and 4 will kick in, contradicting the fact that $(p, x)$ is a fixed point. Therefore $B_i \subseteq G_i^*$.  Clearly, $u$ is the maximum utility that $i$ can derive from a bundle satisfying $\size(i) = 1$ and $\cost(i) \leq 1$. Therefore, $x_i$ is an optimal bundle for $i$. Since $\cost(i) < 1$, by complementarity $\alpha_i = 0$ and hence $x_i$ is a Type A optimal bundle. 

		{\em Henceforth, we will assume that $\cost(i) = 1$.}
		
{\bf Case 2:} Assume that $i$ derives the same utility from all goods $j \in B_i$ and $B_i \subseteq G_i^*$. As in the previous case, $x_i$ is an optimal bundle for $i$ and hence each good in $B_i$ is optimal. If $\exists \ j, k \in B_i$ such that $p_j \neq p_k$, then we get
\[ \alpha_i p_j = u_{ij} - \mu_i = u_{ik} - \mu_i = \alpha_i p_k, \]
implying $\alpha_i = 0$. In this case, $x_i$ is a Type B bundle. 

Otherwise, since $\cost(i) = 1$, $\size(i) = 1$ and all goods in $B_i$ have the same price, each good in $B_i$ has price 1. Now, $x_i$ is either a Type B or a Type C bundle depending on whether $\alpha_i = 0$ or $\alpha_i > 0$. However, the latter condition cannot be discerned from $x_i$; fortunately, it does not matter.
		
{\bf Case 3:} Assume that $i$ derives the same utility from all goods $j \in B_i$ and $B_i  \not  \subseteq G_i^*$. Let $k$ be a good in $B_i$ and let $z$ be a good having price 0. Each good in $B_i$ must be a minimum price good having utility $u_{ik}$, since otherwise Step 5 of $F_i$ will alter the bundle. Since $\cost(i) = 1$, $\size(i) = 1$ and all goods in $B_i$ have the same price, each good in $B_i$ has price 1. 
		
		Let $l$ be a good such that $u_{il} > u_{ik}$; observe that any good in $G_i^*$ is such a good. We will prove that $p_l > 1 = p_k$. Clearly $u_{iz} < u_{ik}$, since otherwise Step 5 will kick in a change the bundle. Hence we have $u_{iz} < u_{ik} < u_{il}$. However, since Step 6 did not kick in, $(u_{il}-u_{ik})(p_k -p_z) \leq (u_{ik}-u_{iz})(p_l -p_k)$. Since $(u_{il}-u_{ik})(p_k -p_z) > 0$, we get that $(p_l -p_k) >0$. Therefore $p_l > p_k = 1$. Hence we can conclude that the optimal bundle for $i$ at prices $p$ is not a Type A or Type B bundle.

		Next, assume for the sake of contradiction that $x_i$ is not an optimal bundle for $i$ at prices $p$; in particular, this entails that the optimal bundle for $i$ is not Type C. Therefore, $i$'s optimal bundle must be Type D and $k$ is a suboptimal good. As argued in Section \ref{sec:optimal}, an optimal Type D bundle must contain a good of price $< 1$ and a good of price $> 1$; let $j$ and $l$ be such goods, respectively. Clearly $u_{iz} < u_{ik} < u_{il}$. Then we have,
		\[ \alpha_i p_j = u_{ij} - \mu_i, \ \ \ \ \alpha_i p_k > u_{ik} - \mu_i \ \ \ \mbox{and} \ \ \  \alpha_i p_l = u_{il} - \mu_i \]
		Subtracting the first from the second and the second from the third we get
		\[ \alpha_i (p_k - p_j) > (u_{ik} - u_{ij}) \ \ \ \  \mbox{and} \ \ \ \  \alpha_i (p_l - p_k) < (u_{il} - u_{ik}) \]
		This gives
		\[ (u_{il}-u_{ik})(p_k -p_j) - (u_{ik}-u_{ij})(p_l -p_k) > 0. \]
		Therefore, Step 6 should kick in, leading to a contradiction. Hence $x_i$ is a Type C optimal bundle.

	{\em Henceforth, we will assume that $\cost(i) = 1$ and $\exists \ s, t \in B_i$ with $u_{is} < u_{it}$.}
		
{\bf Case 4:} Assume that the set $\{u_{ij} \ | \ j \in G \}$ has exactly two elements. Clearly, these utilities must be $u_{is}$ and $u_{it}$. Now, $s$ must be the zero-priced good, since otherwise Step 5 will kick in. Since $\cost(i) = 1$ and $\size(i) = 1$, $p_t > 1$. Again since Step 5 didn't kick in, $s$ and $t$ must be the cheapest goods having utilities $u_{is}$ and $u_{it}$. Therefore, $x_i$ is a Type D optimal bundle.
			  
{\bf Case 5:} Assume that the set $\{u_{ij} \ | \ j \in G \}$ has three or more elements. Since $\size(i) = 1$ and $\cost(i) = 1$, $\exists \ t \in B_i, \ s.t. \ p_t > 1$. Now, any good having utility $u$ must have price $> 1$, since otherwise Step 5 will alter the bundle. Therefore, $x_i$ cannot be a Type A or Type B bundle. Therefore, $\alpha_i > 0$. 

Suppose that $x_i$ is not an optimal bundle. Then there are two cases: that the optimal bundle is Type C or Type D. In the first case, let $k \in G$ be an optimal good; $p_k = 1$. Let $j, l \in B_i$ with $p_j < 1 < p_l$ and at least one of $j$ or $l$ is suboptimal. Clearly, $u_{ij} < u_{ik} < u_{il}$, otherwise Step 5 will kick in. Therefore we have 
		\[ \alpha_i p_j \geq u_{ij} - \mu_i, \ \ \ \ \alpha_i p_k = u_{ik} - \mu_i \ \ \ \mbox{and} \ \ \  \alpha_i p_l \geq u_{il} - \mu_i , \]
with at least one of the inequalities being strict. Therefore, 
\[ (u_{ik}-u_{ij})(p_l -p_k) > (u_{il}-u_{ik})(p_k -p_j) ,\]
and Step 7 should kick in, leading to a contradiction. Hence $x_i$ is a Type C optimal bundle.

Next suppose the optimal bundle is Type D. There are two cases. First, $\exists \ k \in B_i$ such that $k$ is a suboptimal good for $i$ and there are optimal goods $j$ and $l$ with $u_{ij} < u_{ik} < u_{il}$. Then we have
		\[ \alpha_i p_j = u_{ij} - \mu_i, \ \ \ \ \alpha_i p_k > u_{ik} - \mu_i \ \ \ \mbox{and} \ \ \  \alpha_i p_l = u_{il} - \mu_i \]
	As before we get
				\[ (u_{il}-u_{ik})(p_k -p_j) - (u_{ik}-u_{ij})(p_l -p_k) > 0. \]
		Therefore, Step 6 should kick in, leading to a contradiction. 
		
Second, that there is no such good $j \in B_i$. Let $v$ and $w$ be optimal goods with the smallest and largest utilities for $i$. Then all suboptimal goods in $B_i$ have either less utility than $u_{iv}$ or more utility than $u_{iw}$. Suppose there are both types of goods, say $j$ and $l$, respectively. Then Step 7 should kick in with the triple $j, v, l$. Else there is only one type, say $j$ with $u_j < u_v$. Then $\exists \ l \in B_i$ with $p_l > 1$. Now, Step 7 should kick in with the triple $j, v, l$. In the remaining case, $\exists \ j, l \in B_i$ with $p_j < 1$ and $u_{il} > u_{iw}$. Now, Step 7 should kick in with the triple $j, w, l$.   

The contradictions give us that $x_i$ does not contain a suboptimal good and is hence a Type D optimal bundle.
\end{proof}

Lemmas \ref{lem.two} and \ref{lem.bundle} give:

\begin{theorem}
	\label{thm.FIXP}
	The Hylland-Zeckhauser scheme is in FIXP.
\end{theorem}

	\section{Discussion}
\label{sec.discussion}

The main problem remaining is to determine if the HZ scheme is FIXP-hard. Another obvious open problem is to obtain efficient algorithms for computing approximate equilibria, suitably defined. Beyond this, generalizations and variants of the HZ scheme deserve attention, most importantly to two-sided matching markets \cite{Echenique2019constrained}.  

Encouraged by success on the unit case, we considered its generalization to the case of $\{0, {1 \over 2}, 1\}$ utilities. However, resolving whether this case always has a rational equilibrium is quite non-trivial and we leave it as an open problem. Furthermore, it will not be surprising if even this case is intractable; resolving this is a challenging open problem.

	\section{Acknowledgements}
\label{sec.ack}

We wish to thank Federico Echenique, Jugal Garg, Tung Mai and Thorben Trobst for valuable discussions, and Richard Zeckhauser for providing us with the Appendix to his paper \cite{hylland}. 
	
	\bibliographystyle{alpha}
	\bibliography{refs}
\end{document}